\documentclass[letterpaper, 10 pt, conference]{ieeeconf}  

\IEEEoverridecommandlockouts                              
\usepackage{mathtools}
\usepackage[inkscapeformat=eps]{svg}
\usepackage[usenames,dvipsnames]{xcolor}
\usepackage{listings}
\usepackage[noEnd=true]{algpseudocodex}
\usepackage{algorithm}
\usepackage{enumerate}
\usepackage{footnote}
\usepackage{flushend}

\usepackage{hyperref}
\hypersetup{
    colorlinks,
    linkcolor={red!50!black},
    citecolor={blue!50!black},
    urlcolor={blue!80!black}
}

\usepackage{caption}
\DeclareCaptionFormat{custom}
{%
  {\footnotesize#1#2 #3}
}
\usepackage{gould_macros}

\newcommand{\eg}{\emph{e.g.}}
\newcommand{\smallconc}[2]{\begin{bsmallmatrix} #1 \\ #2 \end{bsmallmatrix}}

\newcommand{\IR}{\mathbb{IR}}

\newcommand{\calH}{\mathcal{H}}
\newcommand{\calI}{\mathcal{I}}

\newcommand{\calN}{\mathcal{N}}
\newcommand{\calO}{\mathcal{O}}

\newcommand{\bfw}{\mathbf{w}}

\newcommand{\bfI}{\mathbf{I}}

\newcommand{\sfG}{\mathsf{G}}

\newcommand{\ul}[1]{\underline{#1}}

\newcommand{\ulw}{\ul{w}}
\newcommand{\ulx}{\ul{x}}
\newcommand{\uly}{\ul{y}}

\newcommand{\ol}[1]{\overline{#1}}

\newcommand{\olw}{\ol{w}}
\newcommand{\olx}{\ol{x}}
\newcommand{\oly}{\ol{y}}

\usepackage{amsthm}
\theoremstyle{definition}
\newtheorem{theorem}{Theorem}
\newtheorem{definition}{Definition}
\newtheorem{proposition}{Proposition}
\newtheorem{example}{Example}

\newtheorem{lemma}{Lemma}

\theoremstyle{remark}
\newtheorem{remark}{Remark}

\newcommand{\revthree}[1]{{#1}}
\newcommand{\revsixteen}[1]{{#1}}
\newcommand{\revtwentyone}[1]{{#1}}
\newcommand{\revtwentytwo}[1]{{#1}}
\newcommand{\edit}[1]{{#1}}

\newcommand{\IH}{\mathcal{I}_\mathcal{H}}
\newcommand{\IHLP}{\mathcal{I}_\mathcal{H}^{\rm LP}}
\newcommand{\IHSS}[1]{\mathcal{I}_{\mathcal{H}, #1}^{\rm SS}}
\newcommand{\kw}[1]{\lstinline!#1!}
\newcommand\perm[2]{\prescript{#1}{}P_{#2}}
\newcommand{\basisvec}[1]{\revsixteen{\delta_{#1}}}
\newcommand{\sampleparammat}{\revtwentyone{\Gamma}}

\definecolor{mygreen}{rgb}{0.1, 0.7, 0.1}

\makeatletter
\newcommand{\removelatexerror}{\let\@latex@error\@gobble}
\makeatother

\title{\LARGE \bf
Automatic and Scalable Safety Verification using Interval Reachability with Subspace Sampling
}

\author{Brendan Gould, Akash Harapanahalli, and Samuel Coogan$^{1}$
\thanks{*This work was supported in part by the National Science Foundation under awards \#2219755 and \#2333488 and by the Air Force Office of Scientific Research under Grant FA9550-23-1-0303.}
\thanks{$^{1}$The authors are affiliated with the School of Electrical and Computer Engineering, Georgia Institute of Technology,
        Atlanta, GA 30332, USA
        {\tt\small \{bgould6, aharapan, sam.coogan\}@gatech.edu}.%
}}

\begin{document}

\maketitle
\thispagestyle{empty}
\pagestyle{empty}

\begin{abstract}

Interval refinement is a technique for reducing the conservatism of traditional interval based reachability methods  by \emph{lifting} the system to a higher dimension using new auxiliary variables and exploiting the introduced structure \revtwentyone{through a \emph{refinement} procedure}. 
We present a novel\revtwentyone{, efficiently scaling, automatic } refinement strategy based on a subspace sampling argument \revtwentyone{and motivated by reducing the number of interval operations through sparsity}.  
\revtwentyone{
Unlike previous methods, we guarantee that refined bounds shrink as additional auxiliary variables are added. 
This additionally encourages automation of the lifting phase by allowing larger groups of auxiliary variables to be considered. 
}
We implement our strategy in JAX, a high-performance computational toolkit for Python and demonstrate its efficacy on several examples, including regulating a multi-agent platoon to the origin while avoiding an obstacle. 

\end{abstract}

\section{Introduction}
\label{sec:intro}

Engineered systems are subject to a wide variety of uncertainties, including measurement errors, mechanical failures, and unmodeled dynamics.
These uncertainties may cause the realized behavior of the system to deviate significantly from its nominal behavior, which is unacceptable in safety-critical systems. 
To combat this problem, many approaches have been developed to verify the safety of a system in the presence of such uncertainty. 
One such strategy is \emph{reachable set over-approximation}~\cite{XC-SS:22,ARCH-COMP24:CHSNL}, where safety is guaranteed for an entire set containing all possible states a system may attain under admissable uncertainties and disturbances.


The set representation chosen for this over-approximation is of great importance---typically resulting in a tradeoff between bound computation time and accuracy. 
Some commonly used set representations include  zonotopes~\cite{CORA}, constrained zonotopes~\cite{ConstrainedZono:2016}, hybrid zonotopes~\cite{HybridZono,zonoLAB}, polytopes~\cite{harwood_efficient_2016}, \revtwentyone{Taylor } models~\cite{JuliaReach}, ellipsoids~\cite{ellipsoids}, and generalized star sets~\cite{StarSets:2016}.
Intervals~\cite{Jaulin:2001} are a particularly simple (and therefore fast to compute) representation of sets, and multiple approaches~\cite{shen_rapid_2017, jafarpour2024efficient} have been proposed for efficient interval reachable set computation.
The efficiency of interval methods has allowed them to, \eg, be incorporated directly to train neural networks controllers with certified robustness guarantees~\cite{harapanahalli_certified_2024}.
Other reachability approaches have also been proposed, like Hamilton-Jacobi reachability~\cite{HJ-reach:2017}, level set methods~\cite{LevelSet:2000}, and contraction-based reachability~\cite{maidens_reachability_2015}.

The simplicity of interval-based approaches causes them to suffer from excessive over-conservatism through the wrapping effect, \revthree{preventing them from producing useful bounds in many cases.}
\revtwentyone{
Recent work~\cite{shen_rapid_2017} has developed methods to mitigate this conservatism by artificially introducing additional structure to a system through model redundancies, which was used in~\cite{harapanahalli_certified_2024} to study invariant polytopes in neural network controlled systems. 
In Section~\ref{sec:reach}, we recall this process through the lens of a lifted system, and describe how domain knowledge can be used to select which lifted structure to create. 
The new structure can be exploited to \emph{refine} bounds for the lifted system, often giving tighter estimates than standard interval reachability. 

In Section~\ref{sec:sampling_strat}, we present the main contribution of this work: a novel, automated refinement strategy based on sampling certain subspaces. 
Through clever parameterization of these subspaces, we induce a desirable sparsity pattern, which heuristically reduces conservatism by cancelling interval operations performed during refinement.
Furthermore, we explicitly bound the growth of refinement costs in both the size of the original system and the amount of added auxilliary variables with Proposition~\ref{prop:sample_growth}, and show how refined bounds are monotonically shrinking as additional redundancies are introduced in Theorem~\ref{thm:bound_monotonicity}. 
These results together represent a step towards the automation of interval refinement, allowing engineers to quickly test many alternative lifting strategies. 
}

Section~\ref{sec:examples} demonstrates our technique on several examples. 
\revtwentyone{
First, we present a framework that fully automates the lifting and refinement for two dimensional systems. 
Next, we show how lifted systems can use some domain knowledge with our automatic refinement procedure to successfully }compute useful reachable set bounds on concentrations of a chemical reaction and positions of agents in a multi-agent platoon. 
Finally, Section~\ref{sec:conclusion} gives a brief conclusion.


\subsection{Notation}
\label{subsec:notation}

We use $\basisvec{i}$ to denote the $i$th standard Euclidean basis vector.
In all other cases, vector subscripts \edit{denote indexing}: $y_j$ is the $j$th component of the vector $y$. 
\edit{Let $\leq$ be the elementwise order on $\R^n$ so that $x \leq y $ when $x_i \leq y_i$ for every $i$. 
Let $[\uly, \oly] := \{y : \uly \leq y \leq \oly\}$ be a closed interval subset of $\R^n$ and $\IR^n$ denote the set of all $n$-dim intervals. }
The symbol $\perm{m}{n}$ is the set of \revtwentytwo{permutations of $n$ elements from a list of size $m$}, and $\lvert \perm{m}{n} \rvert$ is the size of this set. 

\section{Polytope Reachability}
\label{sec:reach}

In this section, we present an approach for over-approximating the reachable sets of continuous-time systems. 
Our technique is a generalization of the procedure presented in~\cite{shen_rapid_2017}, using interval arithmetic with a refinement step. 

\subsection{Lifting and Refinement}

Consider the nonlinear system
\begin{align}
\label{eq:sys}
    \dot{x} = f(x,w),
\end{align}
for $x\in\R^n$ as the state, $w\in\R^p$ as a disturbance, and continuous $f:\R^n\times\R^p\to\R^n$ locally Lipschitz in $x$ for fixed $w$.
Let $H\in\R^{m\times n}$ \revtwentyone{(with $m > n$) } be full rank and $H^+$ satisfy $H^+ H = \bfI_n$. 
We define the $(H,H^+)$-\emph{lifted system}: 
\begin{equation}
    \label{eq:lsys}
    \dot{y} = Hf(H^+y, w).
\end{equation}
\revtwentytwo{Since $H$ is full rank, it must have at least $n$ linearly independent rows. 
Without loss of generality, we assume that the first $n$ rows of $H$ form an invertible block $H_V$ (allowing us to recover $x$ from $y$) and call the remaining block $H_A$.}
\revtwentyone{$H_V$ acts as a state transformation on the \emph{base states} $x \in \R^n$, $H_A$ is another transformation giving a vector of \emph{auxiliary states}, and their concatenation gives the \emph{lifted states} $y \in \R^m$. }

By~\cite[Proposition 11]{harapanahalli_certified_2024}, the linear subspace $\mathcal{H} = \{Hx : x \in \R^n\}$ is forward invariant for the lifted dynamics~\eqref{eq:lsys}.
Next, we recall the definition of an interval refinement operator~\cite{shen_rapid_2017}, which will allow us to exploit this information.

\begin{definition}[Interval refinement operator {\cite{shen_rapid_2017}}]
    Given $\mathcal{S} \subseteq \R^m$, we say $\mathcal{I}_\mathcal{S}: \IR^m \to \IR^m$ is an \emph{interval refinement operator} if for every interval $[\uly, \oly] \in \IR^m$ we have 
    \begin{equation}
        \label{eq:refinement_op}
        \mathcal{S} \cap [\uly, \oly] \subseteq \mathcal{I}_\mathcal{S}([\uly, \oly]) \subseteq [\uly, \oly].
    \end{equation}
\end{definition}

In our context, $\IH$ tightens known bounds on lifted system trajectories; if $y(t) \in [\uly, \oly]$, then also $y(t) \in \IH([\uly, \oly])$. 
$\IH([\uly, \oly])$ is often significantly smaller than $[\uly, \oly]$ itself, improving bounds. 
Below, we present two refinement operators for $\mathcal{H}$, varying in tightness and computational complexity. 

The tightest refinement operator for a lifted system is 
 $\IHLP := [\underline{\IHLP}([\uly, \oly]), \overline{\IHLP}([\uly, \oly])]$, where 
\begin{align}
    \underline{\IHLP}([\uly, \oly])_j &= \min_{x \in \R^n} \basisvec{j}^\top H x \ \text{ s.t. } \ \uly \le Hx \le \oly, \label{eq:IHLP_lower} \\ 
    \underline{\IHLP}([\uly, \oly])_j &= \max_{x \in \R^n} \basisvec{j}^\top H x \ \text{ s.t. } \ \uly \le Hx \le \oly, \label{eq:IHLP_upper}
\end{align}
for each $j \in \{1, 2, \ldots, m\}$, as originally proved in~\cite{harwood_efficient_2016}.
To refine an interval in $\IR^m$ with $\IHLP$, $2m$ linear programs must be solved. 
As the number of states in the lifted system~\eqref{eq:lsys} grows, this computation becomes burdensome.

To remedy this issue, we consider another refinement operator. 
For any $N \in \N$, let $A \in \R^{N \times m}$ such that $AH = 0$, and define the \emph{sampling strategy} refinement operator:
\begin{equation}
    \label{eq:IHSS}
    \IHSS{A}([\uly, \oly])_j = [\uly_j, \oly_j] \hspace{-3mm} \bigcap_{\begin{smallmatrix} i \in \{1, \dots, N\},  \\ A_{i, j} \neq 0 \end{smallmatrix}} \hspace{-3mm}- \frac{1}{A_{i, j}} \sum_{k \ne i} A_{i, k}[\uly_k, \oly_k],
\end{equation}
for all $j \in \{1, 2, \ldots, m\}$.
By~\eqref{eq:refinement_op}, if $\calH \cap [\uly, \oly] = \emptyset$, then~\eqref{eq:IHSS} is vacuous and we may choose any output, \eg, the singleton set $\IHSS{A}([\uly, \oly]) = [\uly, \uly]$.
\revthree{This generalizes a refinement operator proposed in~\cite{shen_rapid_2017} and used in~\cite[Appendix 3]{harapanahalli_certified_2024} by allowing $A$ to be any matrix in the left null space of $H$ rather than just a basis.} 
\revsixteen{Our contribution is a clever parameterization of $A$, allowing us to theoretically characterize both the complexity (Proposition~\ref{prop:sample_growth}) and tightness (Theorem~\ref{thm:bound_monotonicity}) of~\eqref{eq:IHSS}.}


Though the bounds provided by $\IHSS{A}$ may not be as tight as those of $\IHLP$, they are computed by simple arithmetic operations and therefore much faster.
In Section~\ref{sec:sampling_strat}, we discuss a principled approach for selecting $A$ and show that it exhibits desirable growth properties. 

\subsection{Embedding Systems}

As in the previous section, let $\dot{y} = g(y,w) := Hf(H^+y,w)$ be an $(H,H^+)$-lifted system, and let $\calI_\calH$ be a refinement operator on $\calH \subseteq\R^m$. 
An \emph{inclusion function} $\sfG$ for $g$ is a mapping $\sfG = [\ul\sfG,\ol\sfG]:\IR^{n}\times\IR^{m}\to\IR^n$ satisfying $g(x,w) \in \sfG([\ulx,\olx],[\ulw,\olw])$ for all $x\in[\ulx,\olx]$, $w\in[\ulw,\olw]$, or
\begin{align*}
    \ul\sfG([\uly,\oly],[\ulw,\olw]) \leq g(y,w) \leq \ol\sfG([\uly,\oly], [\ulw,\olw]).
\end{align*}
Inclusion functions are easily constructed through composition of intermediate operations~\cite{Jaulin:2001}.
For instance, \verb|immrax|~\cite{harapanahalli_immrax_2024} is an efficient toolbox for JAX that automatically constructs inclusion functions for general dynamics. 
\revtwentyone{%
Throughout, we use a natural inclusion function $\sfG$~\cite{Jaulin:2001} which \texttt{immrax} obtains by algorithmically replacing nodes in a computation tree with their interval extensions~\cite{harapanahalli_immrax_2024}.
}

Using $\sfG$, we define a \emph{lifted embedding system} by composing the interval refinement operator with the inclusion function evaluated on \revtwentytwo{each lower ($[\uly,\oly_{i:\uly}]$) and upper ($[\uly_{i:\oly},\oly]$) face of the hyperrectangle 
$[\uly,\oly]$\footnote{ For $x, y \in \R^n$, $x_{i:y} \in \R^n$ is the vector \emph{replacing} the $i$th component of $x$ with that of $y$: $(x_{i:y})_j = x_j$ for every $j \ne i$ and $(x_{i:y})_i = y_i$. }, for every $i\in\{1,\dots,m\}$,}
\begin{subequations} \label{eq:lifted_embsys}
\begin{align} 
    \dot{\uly}_i &= \ul\sfG_i(\IH([\uly,\oly_{i:\uly}]),[\ulw,\olw]), \\
    \dot{\oly}_i &= \ol\sfG_i(\IH([\uly_{i:\oly},\oly]),[\ulw,\olw]).
\end{align}
\end{subequations}
\revtwentyone{
Under some mild regularity assumptions, it follows from~\cite[Thm. 2]{scott_barton_reach:2013} that for any $x_0\in\{x\in\R^n : \uly_0\leq Hx\leq\oly_0\}$, and $t\mapsto\bfw(t)\in[\ul\bfw(t),\ol\bfw(t)]$, we have
\begin{align*}
    x(t) \in \{x\in\R^n : \uly(t)\leq Hx\leq\oly(t)\},
\end{align*}
where $t\mapsto x(t)$ is the trajectory of~\eqref{eq:sys} with initial condition $x_0$ and input map $\bfw$, and $t\mapsto [\uly(t), \oly(t)]$ is the trajectory of any $(H,H^+)$-lifted embedding system~\eqref{eq:lifted_embsys} from initial condition $[\uly_0, \oly_0]$ under input map $[\ul\bfw, \ol\bfw]$.
}




\begin{remark}
\revtwentyone{
In~\cite{shen_rapid_2017}, the authors tighten bounds by re-expressing dynamics in terms of the lifted variables, which is represented in our lifted system framework through different left inverses $H^+$~\cite[Remark 15]{harapanahalli_certified_2024}.
If $H = \begin{bsmallmatrix} 1 & 0 \\ 0 & 1 \\ 1 & 1 \end{bsmallmatrix}$, the choice $H^+ = \begin{bsmallmatrix} 1 & 0 & 0 \\ 0 & 1 & 0 \end{bsmallmatrix}$ keeps the first $2$ components of the lifted dynamics the same, while $H^+ = \begin{bsmallmatrix} 1 & 0 & 0 \\ -1 & 0 & 1\end{bsmallmatrix}$ replaces $x_2$ with $-x_1 + y_3$.
However, we note that the choice of $H^+$ does not directly result in the term cancellations also proposed in~\cite{shen_rapid_2017}.
For the rest of this paper, we use $H^+ = \begin{bmatrix} H_V^{-1} & 0_{n, m-n} \end{bmatrix}.$
}
\end{remark}

\section{Refinement Strategy}
\label{sec:sampling_strat}

In the previous section, we defined the sampling strategy refinement operator $\IHSS{A}$ but left the parameters $N \in \N$ and $A \in \R^{N \times m}$ in~\eqref{eq:IHSS} undetermined. 
There are many suitable $A$ matrices, and importantly each may yield different refined bounds. 
We therefore wish to characterize choices of these parameters that produce good refinements efficiently. 

\revtwentyone{
First, we show that when auxiliary variables are considered individually, the optimal refinement can be parameterized by a single $a \in \R^{m-n}$. 
Combining these vectors for each auxiliary variable gives a basis for the left null space of $H$ with desirable properties that are made clear in Example~\ref{ex:basis}. 
}

\begin{example}[\edit{Basis } computation]
\label{ex:basis}
Given a matrix $H$, the MATLAB code \kw{null(H.').'} uses SVD to compute $L_{\rm M}(H)$ with rows forming a basis for the left null space of $H$. 
Applying $L_{\rm M}$ to two very similar matrices, we see
{
\begin{align*}
    L_{\rm M}\left( \hspace{-0.5mm}
    \begin{bsmallmatrix}
        1 & 0 \\
        0 & 1 \\
        1 & 1
    \end{bsmallmatrix}
    \hspace{-0.5mm} \right) &= \begin{bsmallmatrix}
        -0.5774 & -0.5774 & 0.5774
    \end{bsmallmatrix}, \\ 
    L_{\rm M}\left( \hspace{-0.5mm}
    \begin{bsmallmatrix}
        1 & 0 \\
        0 & 1 \\
        1 & 1 \\
        1 & 0.5 
    \end{bsmallmatrix}
    \hspace{-0.5mm} \right) &= \begin{bsmallmatrix}
  -0.9957 & -0.7071 &  0.4184 & 0.5773 \\ 
   0.0917 & -0.4082 &  0.9082 & -1.0000
    \end{bsmallmatrix} .
\end{align*}
}%
By contrast, the bases computed in Algorithm~\ref{alg:null_basis_construction}, Line~\ref{alg:null_basis_construction:basis} are
{
\begin{align*}
    L\left( \hspace{-0.5mm}
    \begin{bsmallmatrix}
        1 & 0 \\
        0 & 1 \\
        1 & 1
    \end{bsmallmatrix}
    \hspace{-0.5mm} \right) &= \begin{bsmallmatrix}
        -1 & -1 & 1
    \end{bsmallmatrix}, \quad
    L\left( \hspace{-0.5mm}
    \begin{bsmallmatrix}
        1 & 0 \\
        0 & 1 \\
        1 & 1 \\
        1 & 0.5 
    \end{bsmallmatrix}
    \hspace{-0.5mm} \right) = \begin{bsmallmatrix}
  -1 & -1 &  1 & 0 \\ 
   -1 & -0.5 &  0 & 1
    \end{bsmallmatrix} .
\end{align*}
}%
\end{example}

\revtwentyone{
Note that with Algorithm~\ref{alg:null_basis_construction}, the first basis appears as a sub-block of the second (but not with $L_{\rm M})$.
This is the key property that ensures adding auxiliary variables does not worsen refinement (Theorem~\ref{thm:bound_monotonicity}).
Furthermore, the right block of our bases is an identity matrix, forcing sparsity (especially when $m-n$ is large). 
This causes many terms in~\eqref{eq:IHSS} to cancel entirely, reducing the number of interval operations, each of which contribute to the conservatism of the refinement. 

Finally, we generate additional rows for $A$ by combining each \emph{pair} of optimal refinement vectors in $L$. 
This gives an $N$ (and induced complexity of~\eqref{eq:IHSS}) much lower than naive, uniform sampling the left null space of $H$ (Proposition~\ref{prop:sample_growth}).
}



\subsection{Algorithm~\ref{alg:null_basis_construction}: Subspace Sampling}


First, we parameterize \edit{the sample space of $A$}, applying the assumption in~\eqref{eq:IHSS} that $AH = 0$. 

\begin{proposition}
    \label{prop:null_dim_reduction}
    Let $H \in \R^{m \times n}$ be full rank and $L^\top\in\R^{m\times (m-n)}$ be a basis for the left null space $\calN(H^\top)$.
    If $N \in \N$ and $A \in \R^{N \times m}$ with $AH = 0$, there exists a $\sampleparammat \in \R^{N \times (m-n)}$ with $A = \sampleparammat L$.
\end{proposition}

\begin{proof}
    Since $AH = 0$, every row $A_i$, $i \in \{1, 2, \ldots, N\}$ must lie in the left null space of $H$. 
    By assumption, there exist scalars $\sampleparammat_{i, 1}, \ldots \sampleparammat_{i, m-n}$ with $A_i = \sum_{j=1}^{m-n} \sampleparammat_{i, j} L_j$.
    Defining $\sampleparammat_i = \begin{bmatrix} \sampleparammat_{i, 1} & \cdots & \sampleparammat_{i, m-n} \end{bmatrix}$, we have $A_i = \sampleparammat_i L$. 
    Repeating the above argument for all $i$ and letting $\sampleparammat = \begin{bmatrix} \sampleparammat_1 & \cdots & \sampleparammat_N \end{bmatrix}^\top$ gives the desired result. 
\end{proof}

Note that Proposition~\ref{prop:null_dim_reduction} reduces the problem of sampling a subspace of $\mathbb{R}^{m}$ to unconstrained sampling of $\R^{m-n}$. 
Compared to prior work~\cite{shen_rapid_2017}, which chose a \emph{preconditioning vector} from $\R^m$ to generate $A$ (without leveraging the subspace $\cal{H}$), this simplification reduces the dimension of the search space. 
For systems with a large base dimension $n$, \revsixteen{this result is especially important, since only the dimensions corresponding to the auxiliary variables need to be sampled}. 

Next, we observe that~\eqref{eq:IHSS} is invariant under scalings of $A$. 

\begin{lemma}
    \label{lemma:scalar_invariance}
    Let $H \in \R^{m \times n}$ be full rank, and $AH = 0$. 
    Then, for any interval $[\uly, \oly]$ and $\lambda \in \R \setminus \{0\}$, we have 
    \begin{equation}
        \label{eq:scalar_invariance}
        \IHSS{A}([\uly, \oly]) = \IHSS{\lambda A}([\uly, \oly])
    \end{equation}
\end{lemma}

This can be quickly seen by examining the form of~\eqref{eq:IHSS}; the multiplication of every $A_{i, k}$ term within the sum is exactly offset by the overall division by $A_{i, j}$. 

\edit{Lemma~\ref{lemma:scalar_invariance} further reduces } the sample space to just the positive half of a hypersphere in $\R^{m-n}$.
We note that if $m-n=1$, then Lemma~\ref{lemma:scalar_invariance} guarantees the existence of a canonical refinement $\IHSS{A}$ since there is only one \edit{vector in the left null space, up to scalings}. 
Thus, by considering every auxiliary variable individually, we may construct a basis consisting of these canonical refinement vectors for each. 
This is the core idea of Algorithm~\ref{alg:null_basis_construction}.

\begin{figure}[!t]
\removelatexerror
\vspace{-2mm}
\begin{algorithm}[H]
\caption{Subspace-based construction of $A$}
\textbf{Input:} Full rank $H \in \R^{m \times n}$, subspace sample size $s \in \N$
\label{alg:null_basis_construction}
\begin{algorithmic}[1]
    \State $\begin{bmatrix} H_V \\ H_A \end{bmatrix} \gets H$ \Comment{$H_V \in \R^{n \times n}, H_A \in \R^{(m-n) \times n}$}

    \Statex {\textcolor{gray}{\textit{$\triangleright$ Construct basis for left null space}}}
    \State $L = [-H_AH_V^{-1} \ I_{m-n}]$ \Comment{$LH = 0$} \label{alg:null_basis_construction:basis}
    
    \Statex {\textcolor{gray}{\textit{$\triangleright$ Sample pairwise subspaces}}}
    \State Allocate $A_2 \in \R^{s \cdot {\lvert\perm{m-n}{2}\rvert} \times m}$
    \State $W\revtwentyone{_i} \gets \begin{bmatrix} \cos( \frac{i\pi}{s+1}) & \sin( \frac{i\pi}{s+1} ) & 0 & \cdots & 0 \end{bmatrix}$ for $i \in \{1, 2, \ldots, s\}$ \Comment{$W \in \R^{s \times m-n}$}
    \For{$p \in \perm{m-n}{2}$ } \label{alg:null_basis_construction:subspace_count}
        \State $\edit{\Omega_p} \gets$ Apply permutation $p$ to columns of $W$ \label{alg:null_basis_construction:permutation}
        \State Append $\edit{\Omega_p} L$ to $A_2$
    \EndFor
    \State \Return $\begin{bsmallmatrix} L \\ A_2 \end{bsmallmatrix}$
\end{algorithmic}
\end{algorithm}
\vspace{-8mm}
\end{figure}

\edit{Sampling the full domain of Lemma~\ref{lemma:scalar_invariance} uniformly requires exponentially many samples~\cite[Corollary 4.2.11]{vershynin_high-dimensional_2025}. }
However, submanifolds of this hypersphere can be covered with only polynomially many, since the term $\lvert \perm{m-n}{k} \rvert$ grows polynomially with degree $k$ (Figure~\ref{fig:subspace_sample}). 
Algorithm~\ref{alg:null_basis_construction} implements this method up to degree two. 
Though this could be iterated to higher dimensions, we note that in practice pairwise sampling or even just using the basis points is often sufficient. 

\begin{proposition}
    \label{prop:sample_growth}
    For $A \in \R^{N \times m}$ generated by Algorithm~\ref{alg:null_basis_construction}, \revtwentytwo{the number of rows $N$ is worst-case upper bounded by $\calO(s(m-n)^2)$}, and hence the complexity of the refinement operator~\eqref{eq:IHSS} is \revtwentyone{$\calO(sm^2 (m-n)^2)$}.
\end{proposition}

\begin{proof}
    Algorithm~\ref{alg:null_basis_construction} constructs $A$ in two phases. 
    First, it generates a basis $\edit{L} $ (Line~\ref{alg:null_basis_construction:basis}) to parameterize the left null space of $H$ as in Proposition~\ref{prop:null_dim_reduction}.
    By Lemma~\ref{lemma:scalar_invariance}, it suffices consider only \edit{constant } magnitude linear combinations of these basis vectors. 
    Algorithm~\ref{alg:null_basis_construction} considers $s \cdot \lvert \perm{m-n}{2} \rvert$ such vectors, where $s \in \N$ is fixed (Line~\ref{alg:null_basis_construction:subspace_count}). 
    By definition, $\lvert\perm{m-n}{2}\rvert = \frac{(m-n)!}{(m-n-2)!}$, so the size of this second block grows quadratically in $m-n$. 
    Thus, the size $N$ of the overall matrix $A$ grows with the sum of the sizes of its two blocks: $\calO(m-n) + \calO(s(m-n)^2) = \calO(s(m-n)^2)$. 

    \revtwentyone{Finally, note that~\eqref{eq:IHSS} requires computing $N$ bounds for $m$ interval components, each using $m$ multiplications and additions and one division. 
    Therefore, the final complexity is $\calO(m^2N) = \calO(sm^2(m-n)^2)$, as claimed. }
\end{proof}

\revtwentytwo{
Adding base states increases $m$ but keeps $m-n$ constant. 
This causes the complexity of~\eqref{eq:IHSS} to increase quadratically, allowing our technique to scale to systems with many base states. 
When adding auxiliary variables, $m$ increases and $n$ is constant, inducing quartic runtime growth. 
Both cases are significant improvements over the exponential complexity of a dense sampling of the hemisphere in $\R^{m-n}$ (Figure~\ref{fig:subspace_sample}).
}


\subsection{Monotonicity With Respect to Extra Rows in $H$}

The specific basis $L$ in Algorithm~\ref{alg:null_basis_construction} is also important. 
\edit{The $j$-th row $L_j$ corresponds exactly to the canonical vector in the left null space of the matrix $\smallconc{H_V}{(H_A)_j}$, since $L_jH = -\delta_j^TH_AH_V^{-1} H_V + (H_A)_j = 0$.
By isolating each auxiliary variable in its own row, it ensures the corresponding $A$ retains the same rows when new auxiliary variables are added. }
This property is not shared by other methods for generating such a basis, \emph{e.g.} SVD, which prioritize orthonormality. 



\begin{theorem}
    \label{thm:bound_monotonicity}
    Let $H^0 \in \R^{m_0 \times n}$ be full rank, and $[\uly^0, \oly^0] \in \IR^{m_0}$. 
    For any $h \in \R^{1 \times n}$ and $\underline{\gamma}, \overline{\gamma} \in \R$ with $\underline{\gamma} < \overline{\gamma}$, define 
    \begin{equation*}
        H^1 = \begin{bmatrix} H^0 \\ h \end{bmatrix}, \quad [\uly^1, \oly^1] = \left[ \begin{bmatrix} (\uly^0)^\top & \underline{\gamma} \end{bmatrix}^\top, \begin{bmatrix} (\oly^0)^\top & \overline{\gamma} \end{bmatrix}^\top \right].
    \end{equation*}
    Let $A^0 \in \R^{N_0 \times m_0}$, $A^1 \in \R^{N_1 \times m_0 + 1}$ be given by Algorithm~\ref{alg:null_basis_construction} for $H^0$ and $H^1$, respectively. 
    For $j \in \{1, 2, \ldots, m_0\}$,
    \begin{equation}
        \label{eq:non_worsening_bounds}
        \mathcal{I}_{\mathcal{H}^1, A^1}^{\rm SS} ([\uly^1, \oly^1])_j \subseteq \mathcal{I}_{\mathcal{H}^0, A^0}^{\rm SS}([\uly^0, \oly^0])_j.
    \end{equation}
\end{theorem}

\begin{proof}
    We prove monotonicity of the lower bound; the upper case is identical. 
    By~\eqref{eq:IHSS}, $\underline{\mathcal{I}}_{\mathcal{H}^1, A^1}^{\rm SS} ([\uly^1, \oly^1])_j \ge \uly_j^0$ always.
    Thus, if $\underline{\mathcal{I}}_{\mathcal{H}^0, A^0}^{\rm SS}([\uly^0, \oly^0])_j = \uly^0$, then \edit{the proof is trivial}. 

    Else, again by~\eqref{eq:IHSS}, there exists an $i^* \in \{1, 2, \ldots, N_0\}$ with 
    \begin{equation*}
        \underline{\mathcal{I}}_{\mathcal{H}^0, A^0}^{\rm SS}([\uly^0, \oly^0])_j = -\frac{1}{A^0_{i^*, j}} \sum_{k \ne i^*} \min\left\{ A_{i^*, k}^0 \uly_k, A_{i^*, k}^0 \oly_k \right\},
    \end{equation*}
    by definition of interval multiplication~\cite{Jaulin:2001}.
    Note that the $m_0-n$ rows of $L^0$ on Line~\ref{alg:null_basis_construction:basis} in Algorithm~\ref{alg:null_basis_construction} are linearly independent by construction, and therefore form a basis for the left null space of $H^0$. 
    Thus, there exist coefficients $\xi_1, \ldots, \xi_{m_0-n}$ with $A_{i^*}^0 = \sum_{i=1}^{m_0-n} \xi_i L^0_i$.
    Since $A_{i^*}^0$ is a row of $A^0$, these $\xi_i$s are a permuted row of $\edit{\Omega}_p$, as in Line~\ref{alg:null_basis_construction:permutation}.

    Consider the linear combination $A_{i^*}^\prime := \sum_{i = 1}^{m_0 - n} \xi_i L^1_i$.
    Since $\perm{m_0-n}{2} \subset \perm{m_0-n+1}{2}$, $A_{i^*}^\prime$ is a row of $A^1$, say $A_\kappa^1$. 
    By construction, $A_{\kappa, q}^1 = A_{i^*, q}^0$ for all $q \in \{1, 2, \ldots, m_0\}$ and (since $L_{i, m_0+1}^1 = 0$ for all $i \le m_0-n$) $A^1_{\kappa, m_0+1} = 0$. 
    Applying~\eqref{eq:IHSS} once more, 
    \begin{align*}
        \underline{\mathcal{I}}&_{\mathcal{H}^1, A^1}^{\rm SS} ([\uly^1, \oly^1])_j \ge -\frac{1}{A^1_{\kappa, j}} \sum_{k \ne \kappa} \min\left\{ A_{\kappa, k}^1 \uly_k, A_{\kappa, k}^1 \oly_k \right\} = \\
        -&\frac{1}{A^0_{i^*, j}} \sum_{k \ne i^*} \min\left\{ A_{i^*, k}^0 \uly_k, A_{i^*, k}^0 \oly_k \right\} = \underline{\mathcal{I}}_{\mathcal{H}^0, A^0}^{\rm SS}([\uly^0, \oly^0])_j,
    \end{align*}
    as desired.
\end{proof}

These results enable us to apply Algorithm~\ref{alg:null_basis_construction} to a wide variety of systems. 
Since each additional variable incurs a reasonable computation cost (Proposition~\ref{prop:sample_growth}) and can only improve bounds (Theorem~\ref{thm:bound_monotonicity}), good bounds can be achieved by simply constructing a lifted system with lots of auxiliary variables, rather than carefully selecting them from domain knowledge (as in previous work~\cite{shen_rapid_2017}). 


\begin{figure}
    \vspace{1mm} 
    \centering
    \includegraphics[width=0.6\linewidth]{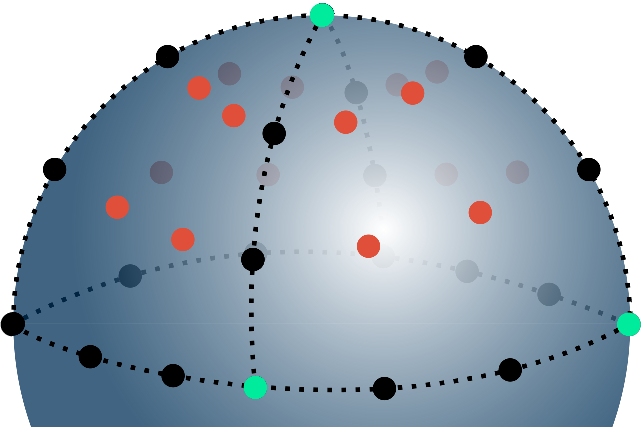}
    \caption{Points sampled on the surface of a hemisphere. 
    Basis samples (corresponding to rows of $L$ in Algorithm~\ref{alg:null_basis_construction}) are shown in green, those along angular subspaces are in black (corresponding to rows of $A_2$ in Algorithm~\ref{alg:null_basis_construction}), and ones distributed over the entire surface in red. 
    \edit{The number of samples needed to densely cover the surface is exponential in dimension.}
    }
    \label{fig:subspace_sample}
\end{figure}

\section{Examples}
\label{sec:examples}

We now demonstrate the efficacy of our approach with several examples\footnote{The source code for these examples is available at \texttt{https://github.com/gtfactslab/Gould\_LCSS2025}.}. 
We include a classic nonlinear system, a chemical reaction, and a multi-agent platoon. 
Additionally, we address uncertainty of many types, including unknown initial conditions, parameter values, or persistent disturbances. 
Our contribution reduces the need to manually tune parameters for each of these domains, successfully applying the same calculations in each. 
Each of the simulations in this section is implemented in \kw{immrax}~\cite{harapanahalli_immrax_2024}, an interval arithmetic and reachability analysis module for the Python framework JAX~\cite{jax2018github} supporting GPU parallelization, Just-in-Time Compilation, and Automatic Differentiation.
They were performed on an Arch Linux system with an Intel i5-12600K CPU, NVIDIA RTX 3050 GPU, and 64GB of RAM. 

\subsection{Refinement Procedure Comparison}
\label{subsec:refine_compare}

Consider the Van der Pol oscillator:
\begin{equation}
    \label{eq:vanderpol}
        \dot{x}_1 = \mu \left(x_1 - \tfrac{1}{3}x_1^3 -x_2\right), \quad \dot{x}_2 = \tfrac{1}{\mu}x_1,
\end{equation}
for some $\mu > 0$.
Fix a number $\ell \in \N$ of auxiliary variables to introduce, and define the matrix $K \in \R^{\ell \times 2}$ whose $i$th row is $\begin{bmatrix} \cos ( \frac{i * \pi}{\edit{\ell}+1} ) & \sin ( \frac{i * \pi}{\edit{\ell}+1} )\end{bmatrix}$ for $i \in \{1, 2, \ldots, \ell\}$.
\revsixteen{Since only the \emph{ratios} of the coefficients of each row of $A$ affect the refinement (6), there is only one degree of freedom to choose these coefficients for two dimensional systems. 
The matrix $K$ produces rows at even intervals in this space. }
Let 
\begin{equation*}
    H = \begin{bsmallmatrix}
        I_2 \\ 
        K
    \end{bsmallmatrix}, \quad 
    H^+ = \begin{bmatrix}
        I_2 & 0_{2 \times \ell}, 
    \end{bmatrix}
\end{equation*}
and consider the $(H, H^+)$-lifted system of~\eqref{eq:vanderpol}. 
\revsixteen{Note that this lifting can be applied automatically to any 2D system.}

\begin{example}
    \label{ex:lp_ss_compare}
    
    We generate the lifted embedding system~\eqref{eq:lifted_embsys} for~\eqref{eq:vanderpol} to over-approximate the reachable set from initial conditions $x_0 \in \left[\begin{bmatrix} 0.9 & -0.1 \end{bmatrix}^\top, \begin{bmatrix} 1.1, 0.1 \end{bmatrix}^\top\right]$ for time $t \in [0, 2\pi]$.
    \edit{We compare the refinements $\IHLP$ given in~\eqref{eq:IHLP_lower}-\eqref{eq:IHLP_upper} and $\IHSS{A}$ based on Algorithm~\ref{alg:null_basis_construction} } \revtwentyone{(with $s=10$)}. 
\end{example}

Table~\ref{table:refinement_comp} compares the average runtime and final bound size of both refinements for multiple values of $\ell$. 
Note that the bound generated by $\IHLP$ is always smaller than that of $\IHSS{A}$ for the same $\ell$. 
However, $\IHSS{A}$ scales better and may generate tighter bounds than $\IHLP$ in a shorter time with a larger $\ell$. 

\begin{table}
\centering
\vspace{1mm} 
\begin{tabular}{|c||c|c||c|c|}
    \multicolumn{3}{c}{Using Refinement $\IHLP$} \\
    \hline 
    $\ell$ & Time (sec) & Final Bound Size \\
    \hline
    2 & $5.78 \pm 1.02 \mathrm{e}{-1}$ & 29.41 \\
    4 & $14.38 \pm 0.269$ & 3.991 \\
    6 & $23.55 \pm 2.63\mathrm{e}{-1}$ & 1.322 \\
    \hline
    \multicolumn{3}{c}{Using Refinement $\IHSS{A}$} \\ 
    \hline
    $\ell$ & Time (sec) & Final Bound Size \\ 
    \hline
    2 & $2.62 \mathrm{e}{-2} \pm 2.98 \mathrm{e}{-3}$ & 29.41 \\ 
    4 & $2.13 \mathrm{e}{-2} \pm 2.48 \mathrm{e}{-3}$ & 4.780 \\ 
    6 & $2.71 \mathrm{e}{-2} \pm 7.81 \mathrm{e}{-3}$ & 1.659 \\ 
    \hline
\end{tabular}
\caption{Comparison of refinements $\IHLP$ and $\IHSS{A}$ as used in Example~\ref{ex:lp_ss_compare}.
Sample size of 10 runs, \revtwentytwo{each numerically integrated from $t=0$ to $t=2\pi$ with a timestep of $0.01$ using the \texttt{tsit5} solver}.}
\label{table:refinement_comp}
\end{table}

The scaling of $\IHSS{A}$ is further supported by the computational capabilities of JAX. 
Computing~\eqref{eq:IHSS} requires repeated evaluations of the same formula. 
The JAX function transformation \kw{vmap} can \emph{vectorize} this, computing many such results in one arithmetic operation. 
Furthermore, the transformation \kw{jit} compiles Python functions into XLA which is executed on a GPU for parallelism. 

%

\subsection{Enzymatic Reaction}
\label{subsec:chem}

Consider the dynamics of an enzymatic reaction from~\cite{shen_rapid_2017}: 
\begin{subequations}
\label{eq:chem}
\begin{align}
\dot{x}_{\mathrm{A}} & =-k_1 x_{\mathrm{A}} x_{\mathrm{F}}+k_2 x_{\mathrm{F}: \mathrm{A}}+k_6 x_{\mathrm{R}: \mathrm{A}^{\prime}}, \\
\dot{x}_{\mathrm{F}} & =-k_1 x_{\mathrm{A}} x_{\mathrm{F}}+k_2 x_{\mathrm{F}: \mathrm{A}}+k_3 x_{\mathrm{F}: \mathrm{A}}, \\
\dot{x}_{\mathrm{F}: \mathrm{A}} & =k_1 x_{\mathrm{A}} x_{\mathrm{F}}-k_2 x_{\mathrm{F}: \mathrm{A}}-k_3 x_{\mathrm{F}: \mathrm{A}}, \\
\dot{x}_{\mathrm{A}^{\prime}} & =k_3 x_{\mathrm{F}: \mathrm{A}}-k_4 x_{\mathrm{A}^{\prime}} x_{\mathrm{R}}+k_5 x_{\mathrm{R}: \mathrm{A}^{\prime}}, \\
\dot{x}_{\mathrm{R}} & =-k_4 x_{\mathrm{A}^{\prime}} x_{\mathrm{R}}+k_5 x_{\mathrm{R}: \mathrm{A}^{\prime}}+k_6 x_{\mathrm{R}: \mathrm{A}^{\prime}}, \\
\dot{x}_{\mathrm{R}: \mathrm{A}^{\prime}} & =k_4 x_{\mathrm{A}^{\prime}} x_{\mathrm{R}}-k_5 x_{\mathrm{R}: \mathrm{A}^{\prime}}-k_6 x_{\mathrm{R}: \mathrm{A}^{\prime}}.
\end{align}
\end{subequations}
We assume the initial concentration of each reagent is given: $x_0 = \begin{bmatrix} 34 & 20 & 0 & 0 & 16 & 0 \end{bmatrix}^\top$, and that the parameter $k$ is known within an order of magnitude: $k \in [\hat{k}, 10 \hat{k}]$, where $\hat{k} = \begin{bmatrix} 0.1 & 0.033 & 16 & 5 & 0.5 & 0.3 \end{bmatrix}^\top$. 

\begin{example}[{\cite[Example 3]{shen_rapid_2017}}]
    \label{ex:chem}
    Using the lifting matrix 
    \begin{gather*}
        K = \begin{bsmallmatrix}
        -0.48 & -0.14 & -0.62 & -0.48 & 0.24 & -0.24 \\ 
        -0.31 & 0.75 & 0.43 & -0.31 & 0.15 & -0.15 \\
        0 & 0 & 0 & 0 & 0.70 & 0.70 
        \end{bsmallmatrix}, \quad
        H = \begin{bsmallmatrix} I_6 \\ K \end{bsmallmatrix} ,
    \end{gather*}
    from~\cite{shen_rapid_2017} with the system~\eqref{eq:chem}, we compute $A$ as in Algorithm~\ref{alg:null_basis_construction} \revtwentyone{with $s=10$ } and use $\IHSS{A}$ to compute reachable set approximations on the time interval $t \in [0, 0.04]$.
    \revsixteen{This $K$ is motivated by chemical properties of the system guaranteeing that each auxiliary variable it introduces is invariant. }
    \revtwentyone{Therefore, we explicitly have $\dot{y}_j = 0$ in~\eqref{eq:lsys} for $j \in \{7, 8, 9\}$.}
\end{example}
The bounds on two of the state variables produced by Example~\ref{ex:chem} are visualized in Figure~\ref{fig:chem}. 
\revtwentyone{
We recover similar bounds to~\cite[Figure 5]{shen_rapid_2017}.
Our bounds rely on Algorithm~\ref{alg:null_basis_construction} and expressing invariance in the lifted system, rather than the radius based preconditioning and symbolic simplifications performed in that work. 
}

\begin{figure}
\vspace{1mm} 
    \centering
    \includesvg[width=\linewidth]{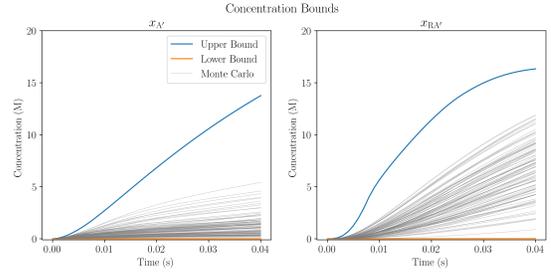} 
    \caption{State bounds for $x_{\mathrm{A}^\prime}$ (left) and $x_{\mathrm{RA}^\prime}$ (right) from Example~\ref{ex:chem}. 
    Gray lines are Monte-Carlo sampled real solutions.  
    \revtwentytwo{The bounds and each sample trajectory were integrated with the \texttt{tsit5} solver and a timestep of $10^{-3}$.}
    }
    \label{fig:chem}
\end{figure}

%

\subsection{Multi-Agent Platoon}
\label{subsec:platoon}
Consider the problem of stabilizing a double integrator subject to an additive input disturbance $w_x^1, w_y^1 \in [-0.01, 0.01]$ while avoiding an obstacle:
\begin{equation}
\label{eq:platoon_leader_dynamics}
\begin{aligned}
    \dot{p}_x^1 &= v_x^1, & \dot{p}_y^1 &= v_y^1, \\ 
    \dot{v}_x^1 &= \sigma(u_x, u_{\mathrm{max}}) + w_x^1, & \dot{v}_y^1 &= \sigma(u_y, u_{\mathrm{max}}) + w_y^1.
\end{aligned}
\end{equation}
Here, $\sigma(\cdot, \cdot)$ is the \emph{softmax} function; we choose $u_{\mathrm{max}}=5$.
For initial state $p_x^1 = 8$, $p_y^1 = 7$, $v_x^1 = -\sqrt{3}$, and $v_y^1 = -1$, we use the Python library \kw{casadi} to generate feedforward inputs $u_{ff} := \{\begin{bmatrix} u_x(t) & u_y(t) \end{bmatrix}^\top\}_{t \in [t_0, t_f]}$ that minimize a quadratic form of position and control. 
Obstacle avoidance (with an additional 30\% padding) is a hard constraint. 

This system can be easily extended to a platoon of arbitrary length $P \in \N$ by appending states with dynamics
\begin{equation}
\label{eq:platoon_follower_dynamics}
\begin{alignedat}{2}
    \dot{p}_x^i &= v_x^i, && \dot{p}_y^i = v_y^i, \\
    \dot{v}_d^i &= k_p \left( p_d^{i-1} - p_d^i - r \frac{v_d^{i-1}}{\lVert v^{i-1} \rVert} \right) && + k_v(v_d^{i-1}-v_d^i) + w_d^i, \\
\end{alignedat}
\end{equation}
for all $i \in \{2, \ldots, P\}$ and $d \in \{x, y\}$. 
Here, each additional agent is implementing PD control to track the proceeding one with parameters $k_p = k_v = 5$.
\revtwentytwo{To avoid collisions, followers do not track the exact position of the proceeding agent, but aim to follow at a distance of $r = 0.5$}.
The follower agents start with the same velocity as the leader and offset position $p_x^i = p_x^1 + 0.2 \sqrt{3} (i-1)$, $p_y^i = p_y^1 + 0.2 (i-1)$.
Finally, the control error of the follower agents is assumed to be similarly bounded: $w_x^i, w_y^i \in [-0.01, 0.01]$.

\begin{example}
    \label{ex:platoon}
    We introduce the lifting matrix
    \begin{gather*}
        L = \begin{bsmallmatrix}
            \basisvec{4(i-1) + 1}^\top + \basisvec{4(i-1)+3}^\top \\
            \basisvec{4(i-1) + 2}^\top + \basisvec{4(i-1)+4}^\top
        \end{bsmallmatrix}, \quad H = \begin{bsmallmatrix} I_{4P} \\ L_1 \\ \vdots \\ L_P \end{bsmallmatrix}, 
    \end{gather*}
    to give the $(H, H^+)$ lifted system for~\eqref{eq:platoon_leader_dynamics}-\eqref{eq:platoon_follower_dynamics}.
    \revsixteen{%
    Leveraging the insight from~\cite[Example 2]{harapanahalli_certified_2024}, it can be advantageous to introduce halfspaces correlating position and velocity terms.
    }
    Constructing $A$ as in Algorithm~\ref{alg:null_basis_construction} \revtwentyone{with $s=10$}, we use the lifted embedding system~\eqref{eq:lifted_embsys} with refinement $\IHSS{A}$ to \edit{bound } the platoon's reachable sets on the time interval $t \in [0, 3]$.
\end{example}

The reachable set bounds computed in Example~\ref{ex:platoon} are visualized in Figure~\ref{fig:platoon} for $P=6$.
Interval refinement tightens these bounds and allows us to verify that every agent in the platoon will avoid the obstacle. 
\begin{figure}
\vspace{1mm} 
    \centering
    \includesvg[width=\linewidth]{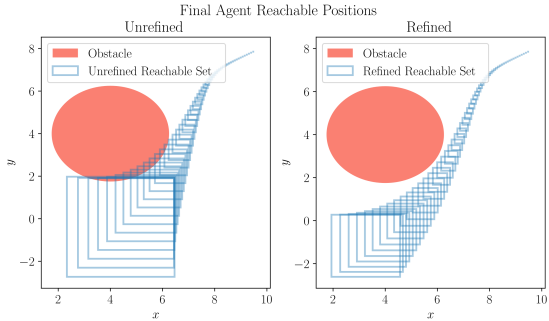}
    \caption{Interval bounds on the position of the final agent in the platoon of Example~\ref{ex:platoon} ($P=6$).
    Bounds produced \emph{without} a refinement operator are shown on the left, and cannot be verified as safe, since they intersect the obstacle. 
    Bounds produced using the refinement operator~\eqref{eq:IHSS} with Algorithm~\ref{alg:null_basis_construction} (shown on the right) are tightened, no longer intersect the obstacle, and verify the safety of the platooned system.  }
    \label{fig:platoon}
\end{figure}

This example demonstrates the scalability of our method. 
The platooned system contains $4P$ base states, and $2P$ auxiliary states are added for a total of $6P$ lifted states. 
Despite the inherent high-dimensionality of multi-agent systems, our method is able to efficiently produce useful bounds on reachable sets. 
Table~\ref{table:platoon_scaling} lists the time taken to compute embedding trajectories for both the base and lifted systems across several values of $P$, characterizing runtime growth in platoon size. 
\revsixteen{%
We additionally compare to zonotopic bounds produced by the CORA toolbox~\cite{CORA}.
Initially, our method verifies safety faster. 
Due to the many auxiliary variables being used, our runtime eventually grows too quickly, and becomes slower than the zonotopic method.
}

\begin{table}
\centering
\footnotesize
\begin{tabular}{|c||c|c|c|}
\hline 
$P$ & Unrefined & Refinement $\IHSS{A}$ & CORA \\ 
\hline 
3 & $1.30 \mathrm{e}{-2} \pm 2.1 \mathrm{e}{-3}$ & $2.92 \mathrm{e}{-2} \pm 2.3 \mathrm{e}{-3}$ & 3.11 \\ 
6 & $1.46 \mathrm{e}{-2} \pm 1.2 \mathrm{e}{-3}$ & $1.97 \pm 3.1 \mathrm{e}{-2}$ & 5.56 \\ 
9 & $2.48 \mathrm{e}{-2} \pm 3.0 \mathrm{e}{-3}$ & $10.3 \pm 0.17$ & 10.57 \\ 
12 & $2.85 \mathrm{e}{-2} \pm 2.0 \mathrm{e}{-3}$ & $35.0 \pm 0.16$ & 12.85 \\
\hline
\end{tabular}
\caption{Scaling of refinement runtime with platoon size. 
Unrefined interval bounds are compared to the subspace sampling refinement strategy of Algorithm~\ref{alg:null_basis_construction} and a zonotope-based method. 
\edit{The refined interval bounds and CORA verify platoon safety; unrefined bounds do not. }
Sample size of 10 runs, \revtwentytwo{each numerically integrated from $t=0$ to $t=3$ with a timestep of $0.01$ using Euler integration}.
\edit{The reported times for CORA were recorded by MATLAB's \texttt{timeit} command.}}
\label{table:platoon_scaling}
\end{table}

\section{Conclusion}
\label{sec:conclusion}

We have presented a novel strategy for parameter conditioning in interval refinement reachability analysis. 
Our method exhibits efficient runtime growth and monotonic bound sizes in the number of auxiliary variables used, enabling its application to a wide variety of systems without the use of domain knowledge. 
Interesting areas for future work include \edit{further automation of the selection of $H$}. 
In particular, a promising approach is to apply the automatic differentiability inherent to JAX with the interval refinement toolbox~\cite{harapanahalli_immrax_2024} to learn an optimal $H$ through gradient descent.


\addtolength{\textheight}{-12cm}   


\bibliographystyle{ieeetr}
\bibliography{gould_ref,extra}

\begin{thebibliography}{10}

\bibitem{XC-SS:22}
X.~Chen and S.~Sankaranarayanan, ``Reachability analysis for cyber-physical systems: Are we there yet?,'' in {\em NASA Formal Methods Symposium}, pp.~109--130, Springer, 2022.

\bibitem{ARCH-COMP24:CHSNL}
L.~Geretti, J.~A.~D. Sandretto, M.~Althoff, L.~Benet, P.~Collins, M.~Forets, S.~Mitsch, C.~Schilling, J.~Tillet, and M.~Wetzlinger, ``Arch-comp24 category report: Continuous and hybrid systems with nonlinear dynamics,'' in {\em Proceedings of the 11th Int. Workshop on Applied Verification for Continuous and Hybrid Systems} (G.~Frehse and M.~Althoff, eds.), vol.~103 of {\em EPiC Series in Computing}, pp.~39--63, EasyChair, 2024.

\bibitem{CORA}
M.~Althoff, ``An introduction to cora 2015,'' in {\em Proc. of the 1st and 2nd Workshop on Applied Verification for Continuous and Hybrid Systems}, pp.~120--151, EasyChair, December 2015.

\bibitem{ConstrainedZono:2016}
J.~K. Scott, D.~M. Raimondo, G.~R. Marseglia, and R.~D. Braatz, ``Constrained zonotopes: A new tool for set-based estimation and fault detection,'' {\em Automatica}, vol.~69, pp.~126--136, 2016.

\bibitem{HybridZono}
T.~J. Bird, H.~C. Pangborn, N.~Jain, and J.~P. Koeln, ``Hybrid zonotopes: A new set representation for reachability analysis of mixed logical dynamical systems,'' {\em Automatica}, vol.~154, p.~111107, 2023.

\bibitem{zonoLAB}
J.~Koeln, T.~J. Bird, J.~Siefert, J.~Ruths, H.~C. Pangborn, and N.~Jain, ``zonolab: A matlab toolbox for set-based control systems analysis using hybrid zonotopes,'' in {\em 2024 American Control Conference (ACC)}, pp.~2513--2520, IEEE, 2024.

\bibitem{harwood_efficient_2016}
S.~M. Harwood and P.~I. Barton, ``Efficient polyhedral enclosures for the reachable set of nonlinear control systems,'' {\em Springer London}, Feb. 2016.
\newblock Accepted: 2016-07-20T19:25:14Z Publisher: Springer London.

\bibitem{JuliaReach}
S.~Bogomolov, M.~Forets, G.~Frehse, K.~Potomkin, and C.~Schilling, ``Juliareach: a toolbox for set-based reachability,'' in {\em Proceedings of the 22nd ACM International Conference on Hybrid Systems: Computation and Control}, pp.~39--44, 2019.

\bibitem{ellipsoids}
A.~B. Kurzhanski and P.~Varaiya, ``Ellipsoidal techniques for reachability analysis,'' in {\em International workshop on hybrid systems: Computation and control}, pp.~202--214, Springer, 2000.

\bibitem{StarSets:2016}
P.~S. Duggirala and M.~Viswanathan, ``Parsimonious, simulation based verification of linear systems,'' in {\em International conference on computer aided verification}, pp.~477--494, Springer, 2016.

\bibitem{Jaulin:2001}
L.~Jaulin, M.~Kieffer, O.~Didrit, and E.~Walter, {\em Applied interval analysis}.
\newblock Springer, 1st edition.~ed., 2001.

\bibitem{shen_rapid_2017}
K.~Shen and J.~Scott, ``Rapid and {Accurate} {Reachability} {Analysis} for {Nonlinear} {Dynamic} {Systems} by {Exploiting} {Model} {Redundancy},'' {\em Computers \& Chemical Engineering}, vol.~106, Aug. 2017.

\bibitem{jafarpour2024efficient}
S.~Jafarpour, A.~Harapanahalli, and S.~Coogan, ``Efficient interaction-aware interval analysis of neural network feedback loops,'' {\em IEEE Transactions on Automatic Control}, 2024.

\bibitem{harapanahalli_certified_2024}
A.~Harapanahalli and S.~Coogan, ``Certified {Robust} {Invariant} {Polytope} {Training} in {Neural} {Controlled} {ODEs},'' Aug. 2024.
\newblock arXiv:2408.01273 [cs, eess, math].

\bibitem{HJ-reach:2017}
S.~Bansal, M.~Chen, S.~Herbert, and C.~J. Tomlin, ``Hamilton-jacobi reachability: A brief overview and recent advances,'' in {\em 56th Conference on Decision and Control (CDC)}, pp.~2242--2253, 2017.

\bibitem{LevelSet:2000}
I.~Mitchell and C.~J. Tomlin, ``Level set methods for computation in hybrid systems,'' in {\em Hybrid Systems: Computation and Control (HSCC)}, pp.~310--323, 2000.

\bibitem{maidens_reachability_2015}
J.~Maidens and M.~Arcak, ``Reachability {Analysis} of {Nonlinear} {Systems} {Using} {Matrix} {Measures},'' {\em IEEE Transactions on Automatic Control}, vol.~60, pp.~265--270, Jan. 2015.
\newblock Conference Name: IEEE Transactions on Automatic Control.

\bibitem{harapanahalli_immrax_2024}
A.~Harapanahalli, S.~Jafarpour, and S.~Coogan, ``immrax: {A} {Parallelizable} and {Differentiable} {Toolbox} for {Interval} {Analysis} and {Mixed} {Monotone} {Reachability} in {JAX},'' Apr. 2024.
\newblock arXiv:2401.11608 [cs, eess, math].

\bibitem{scott_barton_reach:2013}
J.~K. Scott and P.~I. Barton, ``Bounds on the reachable sets of nonlinear control systems,'' {\em Automatica}, vol.~49, no.~1, pp.~93--100, 2013.

\bibitem{vershynin_high-dimensional_2025}
R.~Vershynin, ``High-{Dimensional} {Probability}: {An} {Introduction} with {Applications} in {Data} {Science},'' May 2025.

\bibitem{jax2018github}
J.~Bradbury, R.~Frostig, P.~Hawkins, M.~J. Johnson, C.~Leary, D.~Maclaurin, G.~Necula, A.~Paszke, J.~VanderPlas, S.~Wanderman-Milne, and Q.~Zhang, ``{JAX}: composable transformations of {Python}+{NumPy} programs,'' 2018.

\end{thebibliography}

\end{document}